\documentclass[11pt]{article}
\usepackage{tipa}
\usepackage[all,import]{xy}
\usepackage{graphicx,color}
\usepackage{amsmath}
\usepackage{amsbsy}
\usepackage{amssymb}
\usepackage{cases}
\usepackage{enumerate}
\usepackage{bm}
\usepackage[colorlinks,linkcolor=black,anchorcolor= black,citecolor= black,urlcolor=black]{hyperref}
\usepackage{amsthm}
\usepackage{multirow}
\usepackage{colortbl}
\usepackage{natbib}
\usepackage{subfigure}

\usepackage{tikz}
\usetikzlibrary{trees}

\usepackage{times}
\usepackage{mathptmx}

\definecolor{mygray}{gray}{.9}

\makeatletter
\newcommand*{\rom}[1]{\expandafter\@slowromancap\romannumeral #1@}
\makeatother

\begin{document}

\theoremstyle{definition}
\newtheorem{assumption}{Assumption}
\newtheorem{theorem}{Theorem}
\newtheorem{lemma}{Lemma}
\newtheorem{example}{Example}
\newtheorem{definition}{Definition}
\newtheorem{corollary}{Corollary}
\newtheorem{prop}{Proposition}

\def\letas{\mathrel{\mathop{=}\limits^{\triangle}}}
\def\ind{\begin{picture}(9,8)
         \put(0,0){\line(1,0){9}}
         \put(3,0){\line(0,1){8}}
         \put(6,0){\line(0,1){8}}
         \end{picture}
        }
\def\nind{\begin{picture}(9,8)
         \put(0,0){\line(1,0){9}}
         \put(3,0){\line(0,1){8}}
         \put(6,0){\line(0,1){8}}
         \put(1,0){{\it /}}
         \end{picture}
    }

\def\var{\text{var}}
\def\cov{\text{cov}}
\def\sumN{\sum_{i=1}^N}
\def\summ{\sum\limits_{j=1}^m}
\def\convergeas{\stackrel{a.s.}{\longrightarrow}}
\def\converged{\stackrel{d}{\longrightarrow}}
\def\iidsim{\stackrel{i.i.d.}{\sim}}
\def\indsim{\stackrel{ind}{\sim}}
\def\asim{\stackrel{a}{\sim}}
\def\d{\text{d}}
\def\obs{\text{obs}}
\def\CACE{\textsc{CACE}}
\def\SP{\textsc{SP}}
\def\RD{\textsc{RD}}
\def\RR{\textsc{RR}}
\def\OR{\textsc{OR}}
\def\logit{\text{logit}}
\def\PC{\textsc{PC}}
\def\lgamma{\text{lgamma}}
\def\digamma{\text{digamma}}
\def\trigamma{\text{trigamma}}

\newcommand\luke[1]{
{\textcolor{red}{{\sc Luke}: {\em #1}}}
}

\setlength{\baselineskip}{1.5\baselineskip}

\title{\bf Model-free causal inference of binary experimental data}
\author{Peng Ding\footnote{Department of Statistics, University of California, Berkeley. Address for correspondence: 425 Evans Hall, Berkeley, California, 94720, USA. Email: \url{pengdingpku@berkeley.edu}.}
~and Luke W. Miratrix\footnote{Graduate School of Education and Department of Statistics, Harvard University. 
}
 }
\date{}
\maketitle

\begin{abstract}

For binary experimental data, we discuss randomization-based inferential procedures that do not need to invoke any modeling assumptions. We also introduce methods for likelihood and Bayesian inference based solely on the physical randomization without any hypothetical super population assumptions about the potential outcomes. These estimators have some properties superior to moment-based ones such as only giving estimates in regions of feasible support. Due to the lack of identification of the causal model, we also propose a sensitivity analysis approach which allows for the characterization of the impact of the association between the potential outcomes on statistical inference. 

\bigskip 
\noindent {\bfseries Keywords}: Attributable effect; Average causal effect; Bayesian inference; Completely randomized experiment; Likelihood; Sensitivity analysis.
\end{abstract}

\section{Introduction}

In randomized experiments, the outcome of interest is often binary, in which case the resulting data can be summarized by a $2\times 2$ table. Testing for significant relationships in $2\times 2$ tables has a long history in statistics. \citet{yates1984tests} provided a comprehensive review of this topic, to which Sir David Cox commented, ``discussion of tests for the $2\times 2$ tables can be described as a saga, a story with deep implications.''

In this paper, we give an in-depth discussion of estimating causal effects for those $2\times 2$ tables generated by completely randomized experiments.
Under the potential outcomes framework \citep{neyman::1923, rubin::1974}, each unit has pretreatment potential outcomes corresponding to the potential treatments that unit could receive. 
Finite population causal inference \citep[cf.][]{rosenbaum::2002, imbens::2015book} focuses on the experimental units at hand, and treats all potential outcomes as fixed with the randomization of treatment assignment as the only source of randomness. 
This view allows for weak modeling assumptions and inferential methods that are valid due to the randomization mechanism itself rather than any stated belief in a data generating process. 
Furthermore, by focusing on the finite population, the precision of the usual difference-in-means estimator is greater than those of comparable infinite population models. Unfortunately, the uncertainty of the estimator depends on the association between the potential outcomes, an unidentifiable quantity that can complicate finite population inference \citep{neyman::1923, imbens::2015book}.

Binary outcomes, however, lend enough structure to the problem that these issues can be somewhat circumvented. Because of the discrete nature of the problem, there are only a small number of possible types of units that could exist, which allows for two things. First, we can achieve sharper bounds on the variance of the difference-in-means estimator. Second, we can actually implement model-free likelihood and Bayesian procedures for the usual treatment effects. These estimators have superior performance to the usual moment estimators because they exploit the structure of the problem in order to limit possible estimates to a restricted parameter space.
In particular, the observed data assign zero likelihood outside a well-defined region of possibilities and so such procedures will not return any of these impossible estimates.  Moment estimators, on the other hand, could return such values.

It is well known that the association between the potential outcomes plays an important role in estimating the average causal effect. Different approaches have been used to address this difficulty. Some restrict attention to testing the sharp null hypothesis of zero causal effect for all experimental units \citep{fisher::1935, copas::1973}. 
Some enumerate all possible combinations of the potential outcomes in order to construct exact confidence intervals \citep{rigdon2015randomization, li2015exact}. Some derive bounds on the variances of the estimators over all possible randomizations using the marginal distributions \citep{robins::1988, aronow::2014, ding::2015, fogarty2016discrete}. Some assume non-negative individual causal effects, allowing causal effects to be estimated directly \citep{rosenbaum::2001}, or use structures such as constant shifts \citep{rosenbaum::2002} or dilations to dictate all the individual outcomes \citep{rosenbaum1999reduced}. Recent work on Bayesian inference imputes missing potential outcomes based on their posterior predictive distributions, which requires modeling the potential outcomes as Binomial samples from a hypothetical infinite population \citep{ding::2015}.

The methods we present in this paper are distinct from these; as a ``reasoned basis'' \citep{fisher::1935}, the randomization itself allows for obtaining a likelihood function without any external modeling assumptions. 
To introduce this concept, we begin in Section \ref{sec::mono} (after setting up notation and background in Section \ref{sec::notation}) with the simple case with monotonicity, i.e., no units are negatively affected by treatment. 
Under monotonicity all causal parameters are identifiable, making this process more easily understood. 
We then relax monotonicity in Section \ref{sec::no-mono} to allow for a sensitivity analysis for the variance of the difference-in-means estimator as well as for our new likelihood and Bayesian inference. 
We finally extend these approaches to the attributable effect \citep{rosenbaum::2001} in Section \ref{sec::attributable}, showing that inference of the attributable effect does not depend on the association between potential outcomes. 
We use a real example to illustrate the theory and methods in Section \ref{sec::illustration} and give some concluding remarks in Section \ref{sec::discussion}.  All proofs have been relegated to the Appendix.

\section{Potential Outcomes, Causal Estimands, and Observed Data}
\label{sec::notation}

Consider an experiment with $N$ units, a binary treatment $W$, and a binary outcome $Y$. 
Under the Stable Unit Treatment Value Assumption \citep{rubin::1980}, we 
define $Y_i(w)$ as the potential outcome of unit $i$ under treatment $w$, with $w=1$ for treatment and $w=0$ for control, respectively. Therefore, the potential outcomes forms a $N\times 2$ matrix $\{ (Y_i(1), Y_i(0)) \}_{i=1}^N$, which is sometimes referred to as the ``Science'' \citep{rubin::2005}. With a binary outcome, there are only four types of individuals possible, defined by the pair $(Y_i(1),  Y_i(0) )$ of potential outcomes.  
In particular, if we imagine $Y$ being a binary outcome of survival status,  $(Y_i(1), Y_i(0)) = (1,1)$ would be those who always survived, $(Y_i(1), Y_i(0)) = (0,0)$ would never survive regardless of treatment, and so forth.  The treatment has a positive impact for those with $(Y_i(1), Y_i(0)) = (1,0)$ and a negative impact for those with $(Y_i(1), Y_i(0)) = (0,1)$. 
Because there are only four types of units, the full $N \times 2$ Science Table can be summarized by a $2\times 2$ table formed by the cell counts $N_{jk}=\#\{ i: Y_i(1)=j,Y_i(0)=k \}$ for $j$ and $k=0,1$. This summary Science Table (See Table~\ref{tb::science}) contains all the information about the causal relationship between the treatment and outcome.

Causal effects are defined as comparisons between the potential outcomes. On the difference scale $\tau_i = Y_i(1) - Y_i(0)$ is the individual-level causal effect for unit $i.$ Define $p_w=\sumN Y_i(w)/N$ as the proportion of the potential outcome $Y_i(w)$ being one. Then the average causal effect is defined as
$$
\tau   = \frac{1}{N}\sumN \tau_i = p_1-p_0 = \frac{N_{10} - N_{01}}{N} .
$$
We focus on $\tau$. It is conceptually straightforward to extend our discussion to other causal measures \citep{robins::1988, ding::2015}.

Consider a completely randomized experiment with $N_1$ units receiving treatment and  $N_0$ control. The observed outcomes are deterministic functions of the treatment assignment and potential outcomes, i.e., $Y_i^\obs = W_i Y_i(1) + (1-W_i) Y_i(0)$. Because both the treatment assignments and observed outcomes are binary, there are four observed types of the units classified by $(W_i, Y_i^\obs)$, which gives a different $2\times 2$ table formed by the cell counts $n_{wy}^\obs = \#\{ i: W_i=w, Y_i^\obs = y  \}$ for $w=0,1$ and $y=0,1$. See Table \ref{tb::obs}. This table is distinct from the unknown Science Table \ref{tb::science}.

\begin{table}[t]
\parbox{.45\linewidth}{
\centering
\caption{The summarized Science Table}\label{tb::science}
\begin{tabular}{|c|cc|c|}
\hline
  &  $Y(1)=1$  & $Y(1)=0$ & row sum  \\
\hline
$Y(0)=1$ & $N_{11}$  & $N_{01} $  & $S = N_{11}+N_{01}$   \\
$Y(0)=0$ & $N_{10}$   & $N_{00} $ & $N-S$  \\
\hline
\end{tabular}
}
\hfill
\parbox{.45\linewidth}{
\centering
\caption{The observed Data}\label{tb::obs}
\begin{tabular}{|c|cc|c|}
\hline
  &  $Y^{\obs}=1$  & $Y^{\obs}=0$ & row sum  \\
\hline
$W=1$ & $n_{11}^\obs $  & $n_{10}^\obs $ & $N_1$\\
$W=0$ & $n_{01}^\obs $  & $n_{00}^\obs $ & $N_0$\\
\hline
\end{tabular}
}
\end{table}

Importantly, the potential outcomes, the cell counts $N_{jk}$'s, and the causal estimand $\tau  $ are all fixed. 
The observed cell counts $n_{wy}^\obs$'s, however, are random, but the randomness comes solely from the physical randomization of the treatment assignment.

Classic approaches use the physical randomization to justify exact tests for sharp null hypotheses that fully specify the associated Science Table \citep{fisher::1935, copas::1973, rosenbaum::2002, imbens::2015book}. The sharp null formulation can be further utilized to construct exact confidence intervals for causal effects by inverting randomization tests \citep{rosenbaum::2001,  rigdon2015randomization, li2015exact}. 
We instead evaluated the repeated sampling properties of the estimators of causal effects, and then derived likelihood-based and Bayesian inference without imposing any modeling assumptions whatsoever.

\section{Inference Under Monotonicity}
\label{sec::mono}

We first discuss an important simplifying case where the potential outcomes satisfy monotonicity:
\begin{assumption}
\label{assume::mono} (Monotonicity)
$Y_i(1) \geq Y_i(0)$ for each unit $i$.
\end{assumption}

Monotonicity means that treatment is not harmful to any unit, which rules out the existence of potentially harmed units with $(Y_i(1), Y_i(0)) = (0,1)$, making $N_{01}=0$. 
The case with $Y_i(1) \leq Y_i(0)$ for all $i$ is analogous.
Monotonicity is not refutable based on the observed data as long as the treatment is not harmful to the outcome on average. 
Monotonicity is a strong assumption: it imposes a maximal correlation between the potential outcomes $Y(1) $ and $ Y(0)$, and guarantees the identifiability of all the cell counts $N_{jk}$'s, as described by Proposition~\ref{thm::identi-mono}:

\begin{prop}
\label{thm::identi-mono}
Under monotonicity, $N_{01}=0$ and we can identify (i.e., express parameters as expectations of observed data) the $N_{jk}$'s by
\[
N_{11} = E\left(   \frac{N}{N_0} n_{01}^\obs   \right),\quad
N_{00} = E\left(   \frac{N}{N_1} n_{10}^\obs   \right),\quad 
N_{10} = E\left(   N - \frac{N}{N_0} n_{01}^\obs  - \frac{N}{N_1} n_{10}^\obs  \right).
\]
\end{prop}

Proposition \ref{thm::identi-mono} immediately results in unbiased moment estimators for the $N_{jk}$'s made by plugging in sample moments.  In particular, $\widehat{N}_{10} = N - (N/N_0) n_{01}^{\obs} -  (N/N_1) n_{10}^{\obs}$ and
$$
\widehat{\tau}   = \frac{\widehat{N}_{10}}{N} = 1 - \frac{n_{01}^\obs  }{N_0}   - \frac{ n_{10}^\obs }{N_1}  =   \frac{ n_{11}^\obs }{N_1}  - \frac{n_{01}^\obs  }{N_0} \equiv \widehat{p}_1 - \widehat{p}_0,
$$
where $\widehat{p}_1$ and $\widehat{p}_0$ are the observed proportions of the outcomes being one under treatment and control, respectively. The mean and variance of $\widehat{\tau}  $ then follow by extending \citet{neyman::1923}'s result. 
Monotonicity allows for estimation of the correlation of potential outcomes, giving:

\begin{prop}
\label{thm::moments-mono}
The randomization distribution of $\widehat{\tau}  $ has mean $\tau  $ and variance
\begin{eqnarray}
\var(\widehat{\tau}  ) =  \frac{N}{N-1} \left\{   \frac{ p_1(1-p_1) }{ N_1 }  
+ \frac{  p_0(1-p_0) }{ N_0 }    -   \frac{ \tau  (1-\tau  ) }{ N}   \right\} . \label{eq::var-mono}
\end{eqnarray}
The variance can be estimated by plugging in: 
\begin{eqnarray}
\widehat{V} =   \frac{N}{N-1}  \left\{  \frac{  \widehat{p}_1(1-\widehat{p}_1) }{ N_1 }  
+  \frac{   \widehat{p}_0(1-\widehat{p}_0) }{  N_0 }  -  \frac{ \widehat{\tau}  (1-\widehat{\tau}  ) } { N}  \right\} . 
\label{eq::var-esti-mono}
\end{eqnarray}
Furthermore, $ \left( \widehat{\tau}   - \tau   \right) /  \widehat{V} ^{1/2} \rightarrow \mathcal{N}( 0, 1 )$ in distribution.
\end{prop}

Unlike the classic \citet{neyman::1923} variance expression, all terms in expression \eqref{eq::var-mono} are identifiable. Although a moment estimator with an explicit form can be useful to illustrate sources of information, it might not make full use of the information and can sometimes give estimates outside of the parameter space. An alternative approach is to utilize likelihood and Bayesian inference for the parameters of interest, which restricts our attention to only those values that are possible. Now because $\{ (Y_i(1), Y_i(0))\}_{i=1}^N$ are fixed numbers, we cannot write down the likelihood function based on the usual Binomial models. We can, however, write it down according to an urn model induced by the completely randomized experiment. In particular, view the finite population as a fixed urn containing three types of balls corresponding to the three types of units defined by $(Y(1),Y(0)) = (1,1), (1,0),$ and $(0,0)$. We have $N_{11}$ balls of type $(1,1)$, $N_{10}$ balls of type $(1,0)$, and $N - N_{11} - N_{10}$ balls of type $(0,0)$. We can thus parametrize the population with only $N_{11}$ and $N_{10}$. A completely randomized experiment is then equivalent to drawing $N_1$ balls from this urn to form the treatment arm, and using the remaining $N_0$ balls to form the control arm. This allows for writing down the likelihood based on the observed data as a multivariate Hypergeometric distribution below.

\begin{theorem}
\label{thm::like-mono}
Under monotonicity, the likelihood function of $(N_{10}, N_{11})$ is 
\begin{eqnarray*}
\binom{ N_{11} }{  N_{11} - n_{01}^\obs }
\binom{N_{10}}{ n_{11}^\obs + n_{01}^\obs - N_{11} }  
\binom{ N - N_{10} - N_{11}   }{ n_{10}^\obs  }
\Big / 
\binom{N}{N_1}
,
\end{eqnarray*}  
for any $\left( N_{10}, N_{11} \right)$ in the region
\begin{eqnarray}
\label{eq::region}
 \left\{ \left( N_{10}, N_{11} \right) : n_{01}^\obs \leq N_{11} \leq n_{11}^\obs + n_{01}^\obs \leq N_{10}+N_{11} \leq N - n_{10}^\obs \right\}.
\end{eqnarray} 
The likelihood is zero elsewhere. 
\end{theorem}

There are several curious aspects and consequences to this theorem which we now discuss.
First, before obtaining data, the condition $N_{10}+N_{11}+N_{00}=N$ restricts $(N_{10}, N_{11})$ to take $(N+2)(N+1)/2$ possible values, and $\tau$ can take values $k/N$ for any integer $k\in [-N,N].$
After observing the data, $(N_{10}, N_{11})$ can take only $(n_{11}^\obs+1)(n_{00}^\obs+1)<(N+2)(N+1)/2$ possible values, and there are at most $n_{11}^\obs +  n_{00}^\obs+1$ possible values for $\tau  $, a fact noticed by \citet{rigdon2015exact} from a different perspective.

Second, there are no modeling assumptions on the outcome.
The likelihood is completely driven by the physical randomization. 
This idea is not entirely new: such an urn model was used in \cite{neyman::1923}'s seminal causal inference paper for deriving the unbiased moment estimator and confidence interval for $\tau  $.

Third, the above allows for a maximum likelihood estimate of $\tau$, obtained by maximizing the likelihood over all possible $(N_{10}, N_{11})$. 
This likelihood function can also play a central role in model-free Bayesian inference. For example, if we put a uniform prior on the $(N+2)(N+1)/2$ feasible points of $(N_{10}, N_{11})$, the posterior distribution of $(N_{10}, N_{11})$ concentrates only on the $(n_{11}^\obs+1)(n_{00}^\obs+1)$ points within region (\ref{eq::region}) and is proportional to the likelihood. If we have prior information other than the uniform distribution, we could also incorporate it into our Bayesian inference. Based on the posterior distribution of $(N_{10}, N_{11})$, it is straightforward to obtain the posterior distribution of $\tau  $.

\section{Inference Without Monotonicity}
\label{sec::no-mono}

We next relax the monotonicity assumption.
Without monotonicity, the unknown parameters in the Science Table, $(N_{11}, N_{10}, N_{01}, N_{00})$, are no longer identifiable by the observed data.
This introduces an additional complication from before, but the overall intuition is the same. 
Without identifiability of $(N_{11}, N_{10}, N_{01}, N_{00})$, the sampling variance of $\widehat{\tau}$ cannot be identified by the observed data, 
the likelihood function will be flat over a region with multiple points, and Bayesian inference will be strongly driven by the prior distribution. 
We can, however, weaken monotonicity in such a way that preserves identifiability in a sensitivity analysis approach.
This can also be used to generate estimation regions rather than point-estimates.
Finally, this approach also allows for continued use of the likelihood approach discussed above.

The key insight is that, for a known $N_{01}$, all the cell counts of $N_{jk}$'s are identifiable, allowing us to parameterize our urn model with $(N_{10}, N_{11})$ as before. We therefore choose $N_{01}$ as the sensitivity parameter with $N_{01} = 0$ corresponding to monotonicity. 

We first present some extensions of the previous propositions, and then discuss how to use them for this sensitivity analysis approach to variance estimation.
We also extend the likelihood and Bayesian inference procedure from before.

\begin{prop}
\label{thm::identifiability-no-mono}
When $N_{01}$ is known, we can identify the $N_{jk}$'s by
\begin{eqnarray*}
N_{11} = E\left(  \frac{N}{N_0} n_{01}^\obs - N_{01} \right),\quad
N_{00} = E\left(  \frac{N}{N_1} n_{10}^\obs - N_{01}  \right), \quad
N_{10} = E\left(  N+N_{01}-  \frac{N}{N_0} n_{01}^\obs -   \frac{N}{N_1} n_{10}^\obs \right).
\end{eqnarray*}
\end{prop} 
The above  derives from  the marginal distributions of the potential outcomes imposing weak restrictions on the association.
This restriction comes from the data being binary.

\begin{prop}
\label{thm::bounds-sensitivity}
The number of potentially harmed units, $N_{01}$, is bounded by
\begin{eqnarray}
\max(0, - N\tau  ) \leq N_{01} \leq \min\{   Np_0, N(1-p_1)    \} .
\label{eq::frechet}
\end{eqnarray}
If we assume a non-negative correlation between the potential outcomes, the bounds become
$$
\max(0, - N\tau  ) \leq
N_{01}\leq Np_0(1-p_1).
$$
If we further assume that non-negative average causal effect $\tau \geq 0$, then the bounds become
\begin{eqnarray}
\label{eq::bound-n01}
0\leq N_{01}\leq Np_0(1-p_1) . 
\end{eqnarray}
\end{prop}

The bounds in \eqref{eq::frechet} are the Frech\'et--Hoeffding bounds \citep[cf.][]{nelsen2007introduction} for $N_{01}$ based on the marginal distributions of the potential outcomes. In many realistic cases, it seems plausible to assume a nonnegative association between the potential outcomes. Without loss of generality, we assume that our data have $\widehat{\tau}>0$, and therefore we either assume monotonicity or conduct sensitivity analysis within the empirical range of \eqref{eq::bound-n01}.

\begin{prop}
\label{thm::variance-bounds-non-mono}
With a known $N_{01}$, the variance of $\widehat{\tau}$ is 
\begin{equation}
\var(\widehat{\tau}  ) = 
\frac{N}{N-1} \left\{   \frac{p_1(1-p_1)}{N_1} + \frac{p_0(1-p_0)}{N_0}   -  
\frac{ \tau  (1-\tau  )  }{N} - \frac{2N_{01}}{N^2} \right\} .  \label{eq::var-for-N01}
\end{equation}
The bounds of the above variance over the possible values of $N_{01}$ as delineated by region \eqref{eq::bound-n01} are
\begin{equation*}
\frac{N}{N-1} \left\{   \frac{\frac{N_0}{N} p_1(1-p_1)}{N_1} + \frac{\frac{N_1}{N} p_0(1-p_0)}{N_0}   \right\} 
\leq 
\var(\widehat{\tau}  )  
\leq 
\frac{N}{N-1} \left\{   \frac{p_1(1-p_1)}{N_1} + \frac{p_0(1-p_0)}{N_0}  - \frac{\tau  (1-\tau  )}{N}   \right\}  . 
\end{equation*}
\end{prop}

The upper bound of $\var(\widehat{\tau})$ corresponds to monotonicity, and the lower bound corresponds to uncorrelated potential outcomes.

\subsection{Variance estimation in a sensitivity analysis}

Although $\tau$ depends only on the marginal distributions of the potential outcomes, the variance of $\widehat{\tau}$ depends further on the association between the potential outcomes. \citet{ding::2015} showed that (\ref{eq::var-mono}) is an upper bound for the true sampling variance of $\widehat{\tau}$ without monotonicity. However, this result does not show explicitly the impact of the association between the potential outcomes on the variability of the estimator for $\tau$. 
Proposition~\ref{thm::variance-bounds-non-mono} does.
In particular, we can conduct a sensitivity analysis by varying $N_{01}$ within \eqref{eq::bound-n01} to get a series of variance estimators according to \eqref{eq::var-for-N01}. If we believe that $N_{01}$ is in a specific range, we can take the maximum and minimum of the variances as a range of possible uncertainty estimates. Generally, as $N_{01}$ increases the variance goes down; the most conservative (largest) variance estimate corresponds to monotonicity.

\subsection{Likelihood and Bayesian inference}
The discussion above allows for getting sharper estimates on the variance of the classic moment estimators, as compared to the classic Neyman approach. We can also extend the likelihood approach shown for monotonicity in a similar fashion to obtain estimators restricted to the support of the parameter space. For a fixed $N_{01}$, the likelihood function, based on an urn model with four types of balls, is given by the following theorem:

\begin{theorem}
\label{thm::like-non-mono}
Given a fixed $N_{01}$, the likelihood function for $(N_{10}, N_{11})$ is
\begin{eqnarray}
\label{eq::like-non-mono}
\sum_{x\in \mathcal{F}}  
\binom{N_{11}}{x}
\binom{N_{10}}{ n_{11}^\obs - x }
\binom{N_{01}}{ N_{01}+N_{11}-n_{01}^\obs-x}
\binom{N-N_{11}-N_{10}-N_{01}}{ n_{10}^\obs+n_{01}^\obs+x-N_{01}-N_{11} }
\Big /
\binom{N}{N_1}
,
\end{eqnarray}
where the feasible region of the above summation is $\mathcal{F} = \left\{ x : L \leq x \leq U \right\}$ with
\begin{eqnarray*}
L = \max( 0, n_{11}^\obs-N_{10}, N_{11}-n_{01}^\obs, N_{01}+N_{11}-n_{10}^\obs-n_{01}^\obs    )  ,\\
U = \min( N_{11}, n_{11}^\obs, N_{01}+N_{11}-n_{01}^\obs, N-N_{10}-n_{10}^\obs-n_{01}^\obs ).
\end{eqnarray*}

Note that the $x$ in the sum in \eqref{eq::like-non-mono} represents the number of always-survivors randomized to the treatment group; the formula marginalizes over this to get the overall likelihood. When $N_{01} = 0$, the feasible region of $x$ collapses to the point $x=N_{11} - n_{01}^\obs$, and the likelihood function in Theorem \ref{thm::like-non-mono} reduces to the one in Theorem \ref{thm::like-mono}. The proof of Theorem \ref{thm::like-non-mono} in the Appendix shows that, for fixed $0\leq N_{01} \leq N\widehat{p}_0 (1-\widehat{p}_1)$, the likelihood is zero outside the following region of $(N_{10}, N_{11})$:
\begin{eqnarray}
\label{eq::feasible}
\begin{array}{rll}
\max(0, n_{01}^\obs - N_{01} ) \leq & N_{11} & \leq  \min( n_{01}^\obs + n_{11}^\obs, N-n_{00}^\obs-N_{01} ),\\
0\leq &N_{10} & \leq  N-n_{01}^\obs-n_{10}^\obs,\\
\max(n_{11}^\obs+n_{01}^\obs-N_{01}, n_{11}^\obs)\leq & N_{10}+N_{11} & \leq  N-n_{10}^\obs. 
\end{array}
\end{eqnarray}
\end{theorem}

We can then do a sensitivity analysis to see how the likelihood function and the maximum likelihood estimator change as we increase $N_{01}$. These curves can also be calculated for any estimand of interest as the population is fully specified by $( N_{11}, N_{10}) $, given $N_{01}$. For Bayesian inference, if we impose a uniform prior on $(N_{10}, N_{11})$, the posterior distribution of $(N_{10}, N_{11})$ is proportional to (\ref{eq::like-non-mono}). This immediately gives posterior distributions of $\tau$.

\citet{copas::1973} treated \eqref{eq::like-non-mono} as a likelihood function for $(N_{11}, N_{10}, N_{01})$, and observed its pathological behaviors due to the unidentifiability issue. An alternative Bayesian approach might impose a prior distribution on the sensitivity parameter $N_{01}$. Regardless of the identifiability issue, the posterior distributions of the parameters of interest will always be proper because of finite support. \citet{watson2014complications} gave detailed discussion on Bayesian inference by imposing prior distributions on $(N_{11}, N_{10}, N_{01})$ and making connections to posterior predictive checks \citep{rubin::1984, rubin::1998}. However, inference might then be driven by the prior distribution of $N_{01}$, an unidentifiable parameter from the data. Therefore, we recommend the sensitivity analysis approach in both likelihood and Bayesian inference to explicitly show the impact of the correlation between potential outcomes.

\section{The Attributable Effect and the Treatment Effect on the Treated}
\label{sec::attributable}

In the previous sections, we focused the average treatment effect, which is a fixed parameter depending only on the Science table. In practice, other causal quantities may be of scientific interest. For instance, \cite{rosenbaum::2001} proposed to estimate the effect attributable to the treatment,
$$
A = \sumN W_i \tau_i, 
$$
which is closely related to the average treatment effect on the treated units
$
\tau  ^W = \sumN W_i \tau_i / N_1 = A/N_1.
$
Both causal quantities $A$ and $\tau^W$ depend on the treatment assignment as well as the Science table, and thus they are themselves random variables. Therefore, as \cite{rosenbaum::2001} suggested, we need to extend the traditional concepts of point and interval estimation to point and interval prediction of random variables in frequentists' inference. Because the difference between $A$ and $\tau^W$ is the fixed scaling factor $N_1$, we discuss only inference of the attributable effect $A.$

As shown in the proof of Theorem \ref{thm::like-non-mono}, the attributable effect can be written as
\begin{eqnarray}
A = n_{11}^\obs + n_{01}^\obs - N_{01} - N_{11} = n_{11}^\obs + n_{01}^\obs  - S,
\label{eq::attributable}
\end{eqnarray}
with $S=N_{01} + N_{11}$ defined in Table \ref{tb::science}.
Note that $S$ is a parameter depending on the Science table (see Table \ref{tb::science}).
Formula \eqref{eq::attributable} shows a linear relationship between $A$ and $S$, which makes statistical inference of $A$ simpler via statistical inference of $S$. To be more specific, if we had a point estimator $\widehat{S}$ for $S$, then we would have a point predictor $\widehat{A} = n_{11}^\obs + n_{01}^\obs - \widehat{S}$ for $A$.
Furthermore, if we had an interval estimator $[\widehat{S}_l, \widehat{S}_u]$ for $S$, then we would have an interval predictor $[\widehat{A}_l, \widehat{A}_u]$ for $A$, where $\widehat{A}_l =  n_{11}^\obs + n_{01}^\obs - \widehat{S}_l$ and $\widehat{A}_u=  n_{11}^\obs + n_{01}^\obs - \widehat{S}_u$. 
We can thus separate out and capture the randomness in our target estimand with observed data, reducing the statistical uncertainty to a classic parameter estimation problem.

\subsection{Exact inference}

Randomization induces a Hypergeometric distribution $n_{01}^\obs\sim H_S$, where $H_S$ has probability mass function $P(H_S=h) = \binom{S}{h}\binom{N-S}{N_1-h}/\binom{N}{N_1}$ for $ \max(0, S-N_0) \leq  h \leq \min(S,N_1)$. 
This Hypergeometric distribution depends on the unknown parameter $S$, and we can thus use the number of positive outcomes under control, $n_{01}^\obs$, as our observed statistic for conducting inference on $S$.
Fortunately, inference on $S$ based on the Hypergeometric $n_{01}^\obs$ is a classical statistical problem. For example, we can conduct a series of tests $H_{0s}: S=s$, and calculate the $p$-value for each fixed $s$ by measuring the extremeness of $n_{01}^\obs$ given $S$. A choice of the two-sided $p$-value is
\begin{eqnarray}
\label{eq::pvalue}
p(s) = \sum_{ P(H_s=h)   \leq P(H_s=n_{01}^\obs)  } P(H_s=h) , 
\end{eqnarray}
i.e., the sum of all the probability masses that are smaller than or equal to the probability mass of the observed value of the Hypergeometric random variable. 
This effectively orders the possible values of $H_s$, given $s$, by their likelihood, and the sum in \eqref{eq::pvalue} captures the total probability mass in the tails given this ordering.
The Hodges--Lehmann-type point estimator for $S$ corresponds to the $s$ values that attain the maximum $p$-value \citep{hodges1963estimates, rosenbaum::2002}; the point estimator may not be unique due to discreteness. The $1-\alpha$ interval estimator contains all the $s$ values such that $p(s) > \alpha$.

The choice of two-sided $p$-value in \eqref{eq::pvalue} leads to the same procedure as \cite{rosenbaum::2001} and \cite{rigdon2015randomization}. 
We note, however, that the classical literature on Fisher's exact test also proposed other choices of two-sided $p$-values based on a Hypergeometric random variable \citep[cf.][page 92]{agresti2013categorical}. 
Moreover, we could alternatively directly construct confidence intervals for $S$ based on the Hypergeometric $n_{01}^\obs$ without inverting tests. 
Please see \citet{wang2015exact} for classical methods and recent developments in constructing confidence intervals for Hypergeometric parameters. 
Overall, the relationship \eqref{eq::attributable} allows for constructing different point and interval estimators for $A$ based on different approaches for $S$, of which the previous approaches of \citet{rosenbaum::2001} and \citet{rigdon2015randomization} are special cases. 
Furthermore, to make exact inference of the attributable effect, \cite{rosenbaum::2001} invoked monotonicity, but our discussion above does not.  
The inference with or without monotonicity is the same.

\subsection{Neyman-type repeated sampling evaluation}

A natural estimator for $A$ is $N_1\widehat{\tau}$. The following proposition shows that $N_1\widehat{\tau}$ is an unbiased predictor of $A$, and the mean squared error for this prediction depends only on the marginal distribution of $Y(0)$.

\begin{prop}
\label{thm::A-neyman}
Over all possible randomizations, $E(A - N_1 \widehat{\tau}  ) = 0$ and
\begin{eqnarray}
\var(  A - N_1 \widehat{\tau}   ) =  \frac{NN_1}{N_0} S_0^2=   \frac{N^2N_1}{N_0(N-1)} p_0(1-p_0),
\label{eq::attributable-effect}
\end{eqnarray}
where $S_0^2$ is the finite population variance of the control potential outcome. Therefore, $A$ can be unbiasedly predicted by $N_1\widehat{\tau}  $ with estimated mean squared error $  N^2N_1\widehat{p}_0(1-\widehat{p}_0) / \{ N_0(N-1) \} $.
\end{prop}


Proposition \ref{thm::A-neyman} does not rely on monotonicity. Moreover, the first identity in \eqref{eq::attributable-effect} also holds for general outcomes. Interestingly, the variance formula \eqref{eq::attributable-effect} does not depend on the association between the potential outcomes, which was hinted at by \citet{robins::1988} and \citet{hansen2009attributing}.
In particular, by allowing the target of estimation to vary in a randomized experiment, one can seemingly avoid the unidentifiable issue, but the resulting analysis is then conditional, in some sense, on the realized assignment.

\subsection{Bayesian inference}

Bayesian posterior inference for $A$ is straightforward conditional on the observed data. Because of the linear relationship between $A$ and $S$ in \eqref{eq::attributable}, the posterior distribution of $S=N_{01} + N_{11}$ determines the posterior distribution of $A.$ Therefore, with fixed $N_{01}$ (zero under monotonicity and positive for a sensitivity analysis), obtaining the posterior distribution of $A$ is straightforward once we obtain the posterior distribution of $N_{11}.$

\section{Illustration}
\label{sec::illustration}

We re-analyze the data in \citet[][pp. 191]{rosenbaum::2002} concerning death in the London underground. In the London underground, some train stations have a drainage pit below the tracks. When an ``incident'' happens (i.e., a passenger falls, jumps or is pushed from the station platform), such a pit is a place to escape contact with the wheels of the train. Researchers are interested in the mortality in stations with and without such a pit. In stations without a pit, only $5$ lived out of $21$ recorded ``incidents.''
For ``incidents'' in stations with a pit, $18$ out of $32$ lived. Therefore, the observed data can be summarized by $(n_{11}^\obs, n_{10}^\obs, n_{01}^\obs, n_{00}^\obs) = (18, 14, 5, 16)$, viewing ``pit'' versus ``no pit'' as treatment versus control, and life as the outcome. 
For illustration, we view this data set as from a hypothetical completely randomized experiment, ignoring any issues of confounding.

Under monotonicity, the moment-based estimator is $\widehat{\tau} = 0.324$, i.e., we estimate that the chance of survival is about $32$ percentage points higher for stations with a pit. Using the variance estimator in \eqref{eq::var-esti-mono} we end up with a confidence interval of $[0.106, 0.543]$, which is $13\%$ narrower than the usual one of $[0.072, 0.577]$.  See the first row of Table~\ref{tb::rosenbaum}.

We then conduct a sensitivity analysis on monotonicity by varying the value of $N_{01}$, where $N_{01}=0$ corresponds to monotonicity, $N_{01}=5$ corresponds to independent potential outcomes, and $N_{01}=2$ is a value between these two extreme cases. Rows 2 and 3 of Table~\ref{tb::rosenbaum} show estimates and associated confidence intervals for these two different values of $N_{01}$. They are smaller. If we believe some would be harmed, we are then more certain of the average causal effect.
Our improved variance estimator \eqref{eq::var-esti-mono} and the Bayesian approach with a uniform prior both provide improved inference.
The moment estimator is close to the Bayesian posterior modes, but there is slight shift of $2$ percentage points.

Figure \ref{fg::rosenbaum-post} shows the posterior distributions of $\tau  $ with $N_{01}=0,2$ and $5$. The posterior distribution has a higher peak and lighter tails with larger $N_{01}$. This conforms to the frequentists' property that the variance of $\widehat{\tau}  $ becomes larger when $N_{01}$ get smaller, with monotonicity being the extreme case.

\begin{table}[t]
\centering
\caption{Moment and Bayes estimators with $(n_{11}^\obs, n_{10}^\obs, n_{01}^\obs, n_{00}^\obs) = (18, 14, 5, 16)$. Each of columns 2--4 shows the point estimator, interval estimator and its length.}
\label{tb::rosenbaum}
\begin{tabular}{cccc}
$N_{01}$& Neyman's variance & Improved variance & Bayes  \\
$0$	&$0.324\ [0.072, 0.577]\ 0.505$&  $0.324\ [0.106, 0.543]\ 0.437$ & $0.301\ [0.075, 0.509]\ 0.434$\\
$2$	&same as above						&  $0.324\ [0.119, 0.530]\ 0.411$& $0.301\ [0.075, 0.490]\ 0.415$\\
$5$	&same as above						&  $0.324\ [0.141, 0.508]\ 0.367$& $0.301\ [0.094, 0.472]\ 0.378$
\end{tabular}
\end{table}

Regarding the attributable effect under monotonicity,
the Hodges--Lehmann-type estimator is $9$, $10$ or $11$, and the $95\%$ interval estimate is $[2,16]$. The posterior mode for $A$ is $10$, and the $95\%$ highest probability interval for $A$ is $[1,16]$. Figure \ref{fg::rosenbaum-attri} compares the posterior probabilities and standardized $p$-values for testing $A=a$, showing that they have similar shapes.
The moment estimator for $A$ is $10.38$ with confidence interval $[1.56, 19.20]$. The moment estimator is outside of the range of the parameter because $A$ must be an integer. 
Worse, the associated interval estimate is wider, with an upper limit larger than $n_{11}^\obs = 18$, the maximum possible value of $A$ under monotonicity due to $A=n_{11}^\obs + n_{01}^\obs - N_{11} \leq n_{11}^\obs$.

\begin{figure}[hb]
\centering
\subfigure[Sensitivity analysis for the posterior distribution of $\tau$. Three posterior distributions of $\tau$ correspond to three values of the sensitivity parameter $N_{01}$.]{\label{fg::rosenbaum-post}
\includegraphics[width = 0.8\textwidth]{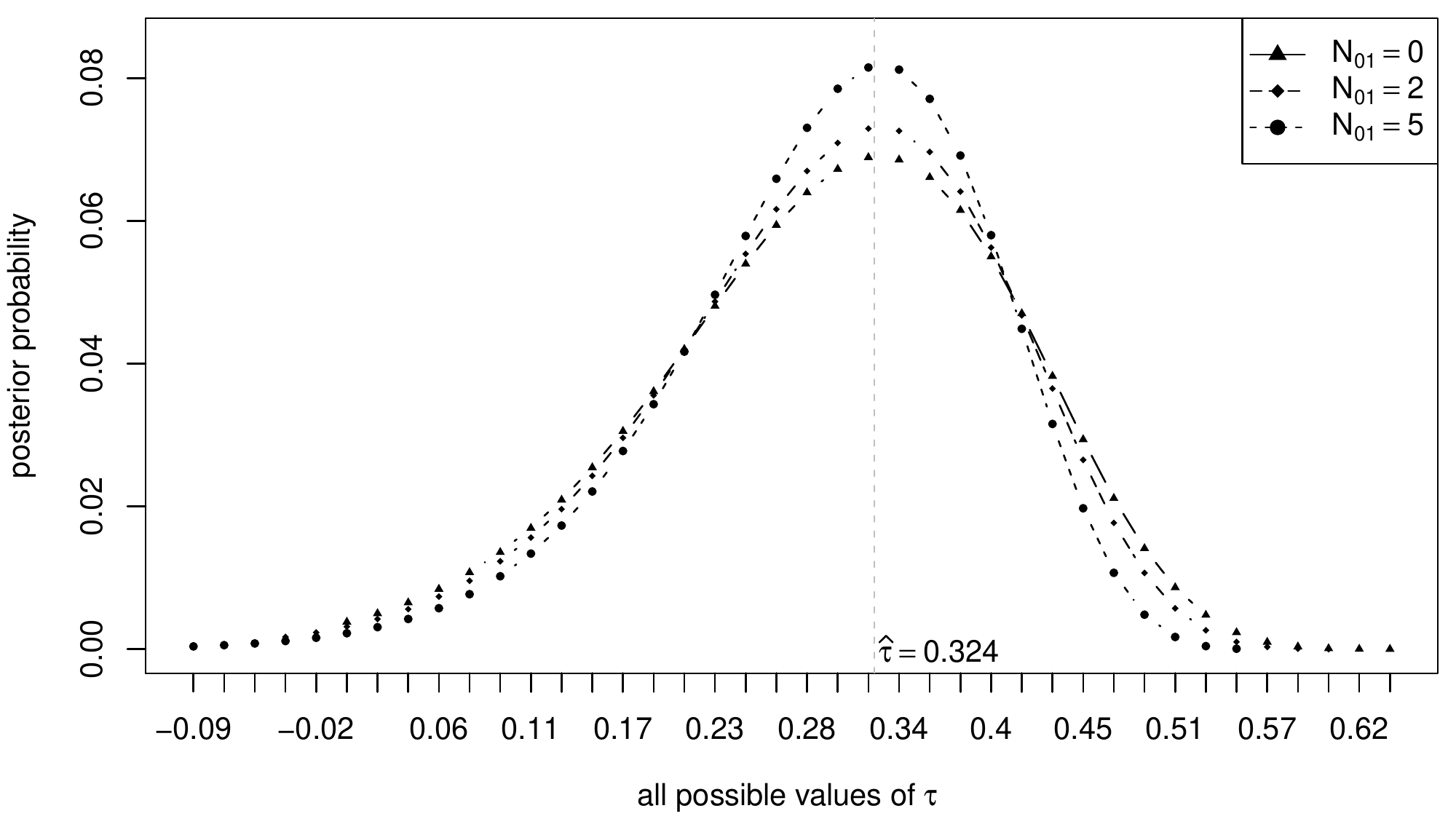}
}
\subfigure[Attributable effect under monotonicity. The $p$-values are standardized to have summation $1$, in order to compare with the posterior distribution.]{\label{fg::rosenbaum-attri}
\includegraphics[width = 0.8\textwidth]{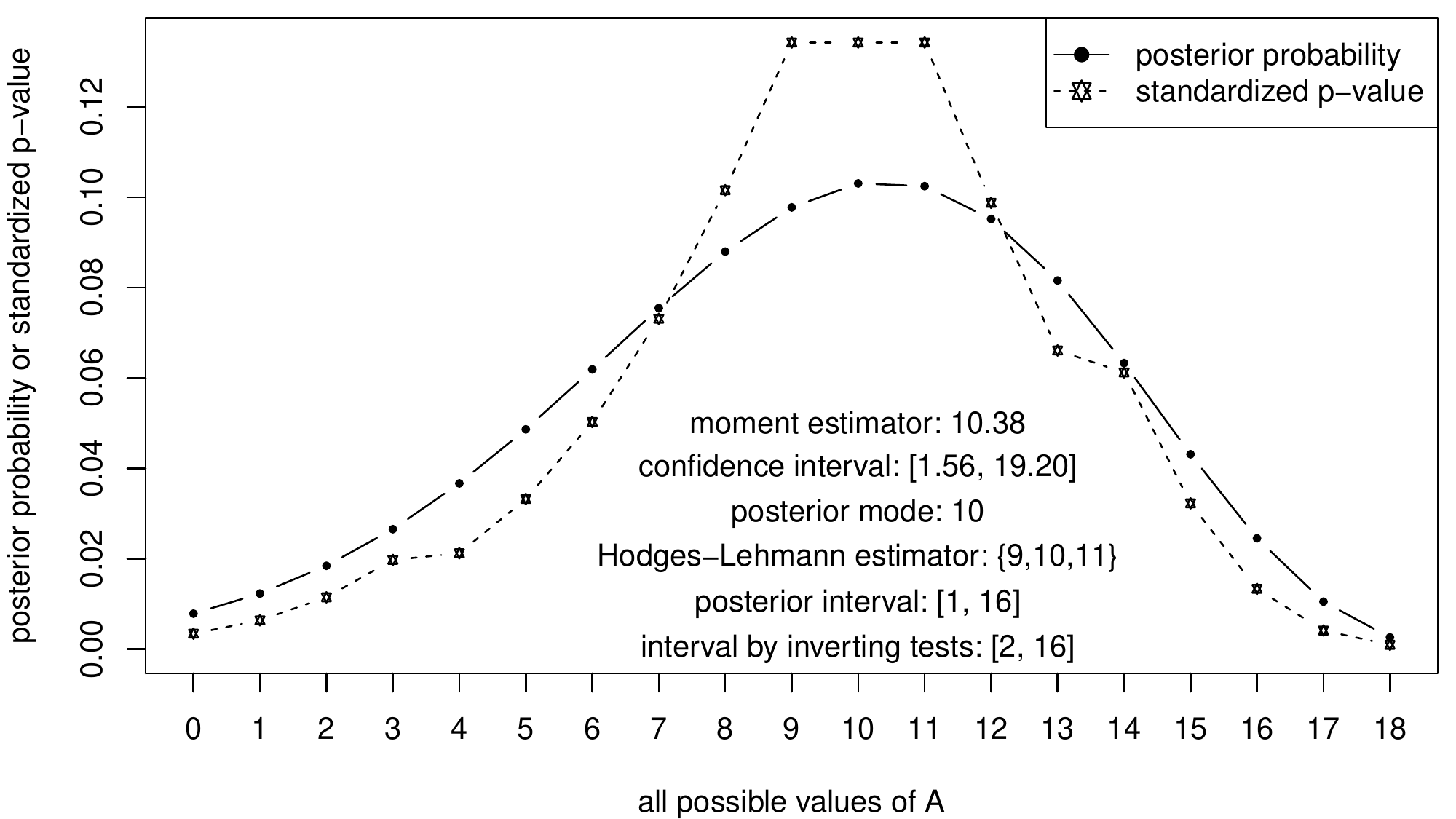}
}
\caption{Example with observed data $(n_{11}^\obs, n_{10}^\obs, n_{01}^\obs, n_{00}^\obs) = (18, 14, 5, 16)$}
\end{figure}

\section{Discussion}
\label{sec::discussion}


For binary experimental data, we proposed several model-free inferential procedures for the average treatment effect and the attributable effect.
We believe demonstrating that likelihood and Bayesian estimation without modeling is possible is a worthwhile proof of concept for an alternate form of thinking about estimation when the assignment mechanism is known.
For further connections and comparisons, see \citet{greenland1991logical}, \citet{ding::2014}, \citet{chiba2015exact}, and \citet{ding::2015}.

Some researchers have proposed randomization-based procedures for causal effects with noncompliance \citep{rubin::1998, imbens2005robust, keele2015randomization}, with general intermediate variables \citep{nolen2011randomization}, and with interference \citep{rosenbaum2012interference, rigdon2015exact}. It is our ongoing work to extend the current approaches to these settings.

\section*{Appendix}

We first prove the theorems.  The propositions follow.
\begin{proof}
[Proof of Theorem \ref{thm::like-mono}]
Under monotonicity, the units with $(W_i,Y_i^\obs) = (1,1)$ are $(1,1)$ or $(1,0)$ units, the units with $(W_i, Y_i^\obs) = (1,0)$ are all $(0,0)$ units, the units with $(W_i, Y_i^\obs) = (0,1)$ are all $(1,1)$ units, and the units with $(W_i, Y_i^\obs)=(0,0)$ are $(0,0)$ or $(1,0)$ units.
Define $N_{bc,w}$ as the number of $(b,c)$ units within observed treatment group $w$. 
Then the observed data allows us to obtain
$$
\begin{array}{lll}
N_{11,1} = N_{11}-n_{01}^\obs, && N_{11,0} = n_{01}^\obs,\\
N_{00,1} = n_{10}^\obs, && N_{00,0} = N_{00}-n_{10}^\obs,\\
N_{10,1} =  n_{11}^\obs - N_{11,1} = n_{11}^\obs+n_{01}^\obs-N_{11},&&
N_{10,0} = N_{10}-N_{10,1}=N_{10}+N_{11}-n_{11}^\obs-n_{01}^\obs.
\end{array}
$$
The above shows that we know the number of each type of unit in both treatment arms, based on the observed counts and the totals $N_{bc}$.
Because all the counts are nonnegative integers, we have the following restriction on $(N_{10}, N_{11})$: 
$$
n_{01}^\obs \leq N_{11} \leq n_{11}^\obs + n_{01}^\obs \leq N_{10}+N_{11} \leq N - n_{10}^\obs.
$$
We can count that there are $ (n_{11}^\obs+1)(n_{00}^\obs+1)$ possible values for $(N_{10}, N_{11})$, and $ (n_{11}^\obs + n_{00}^\obs + 1)$ possible values for $\tau  .$

The completely randomized experiment corresponds to an urn model. 
We have an urn with $N_{11}$  $(1,1)$ balls, $N_{10}$ $(1,0)$ balls, and $N_{00}$ $(0,0)$ balls.
The experiment is that we randomly draw $N_1$ balls without replacement to form the treatment arm and use the remaining balls to form the control arm. We then observe the outcomes.
The above restrictions allows us to determine, based on observed data, the count vector for the three types of balls $(N_{11,1}, N_{10,1}, N_{00,1})$ that we have in the treatment arm, and similarly for control.
Therefore, the probability of obtaining $(N_{11,1}, N_{10,1}, N_{00,1})$ is a multivariate Hypergeometric distribution, given the values of $N_{11}$ and $N_{10}$.  Express this in terms of the observed data to obtain
$$
\binom{ N_{11} }{ N_{11,1}  }
\binom{ N_{10} }{ N_{10,1} }  
\binom{ N_{00} }{ N_{00,1}  }  \Big / \binom{N}{N_1}
=
\binom{ N_{11} }{  N_{11} - n_{01}^\obs }
\binom{N_{10}}{ n_{11}^\obs + n_{01}^\obs - N_{11} }  
\binom{ N - N_{10} - N_{11}   }{ n_{10}^\obs  }
 \Big / \binom{N}{N_1}.
$$
This is the likelihood, a function of $N_{11}$ and $N_{10}$, our parameters.
\end{proof}

\begin{proof}
[Proof of Theorem \ref{thm::like-non-mono}]
Without monotonicity, the observed data classified by $(W_i, Y_i^\obs)$ are mixtures: the observed group $(W_i, Y_i^\obs)=(1,1)$  contains $(1,1)$ and $(1,0)$ units, the observed group $(W_i, Y_i^\obs)=(1,0)$  contains $(0,1)$ and $(0,0)$ units, the observed group $(W_i, Y_i^\obs)=(0,1)$ contains $(1,1)$ and $(0,1)$ units, and the observed group $(W_i, Y_i^\obs)=(0,0)$  contains $(1,0)$ and $(0,0)$ units.
 Assume that $N_{11,1}=x$, we have
$$
\begin{array}{lll}
N_{11,1} = x, && N_{11,0} = N_{11}-x,\\
N_{10,1} = n_{11}^\obs -x, && N_{10,0} = N_{10}+x-n_{11}^\obs,\\
N_{01,1} = N_{01}+N_{11}-n_{01}^\obs-x, && N_{01,0}=n_{01}^\obs+x-N_{11},\\
N_{00,1} = n_{10}^\obs+n_{01}^\obs+x-N_{01}-N_{11}, && N_{00,0} =N-N_{10}-x-n_{10}^\obs-n_{01}^\obs.  
\end{array}
$$
As a byproduct, the attributable effect is $A = N_{10,1}  - N_{01,1}  = n_{11}^\obs  + n_{01}^\obs - N_{01}- N_{11}.$
The above counts must all be non-negative, implying the following inequality on $x$:
$$
\max( 0, n_{11}^\obs-N_{10}, N_{11}-n_{01}^\obs, N_{01}+N_{11}-n_{10}^\obs-n_{01}^\obs    )  
\leq x\leq  
\min( N_{11}, n_{11}^\obs, N_{01}+N_{11}-n_{01}^\obs, N-N_{10}-n_{10}^\obs-n_{01}^\obs ) .
$$
When $N_{01}=0$, the inequality collapses to $x=N_{11}-n_{01}^\obs$, which is coherent with Theorem \ref{thm::like-mono}. The above inequality also imposes the following restrictions on $(N_{10},N_{11})$ for a given value of $N_{01}$ and the observed data:
$$
\begin{array}{ll}
0\leq  N_{11}, &0\leq n_{11}^\obs, \\
0\leq  N_{01}+N_{11}-n_{01}^\obs,&0\leq N-N_{10}-n_{10}^\obs-n_{01}^\obs,\\
n_{11}^\obs-N_{10} \leq N_{11},& n_{11}^\obs-N_{10} \leq n_{11}^\obs,\\ 
n_{11}^\obs-N_{10} \leq N_{01}+N_{11}-n_{01}^\obs,& n_{11}^\obs-N_{10}\leq  N-N_{10}-n_{10}^\obs-n_{01}^\obs,\\
N_{11}-n_{01}^\obs \leq N_{11}, &        N_{11}-n_{01}^\obs \leq n_{11}^\obs, \\
N_{11}-n_{01}^\obs \leq  N_{01}+N_{11}-n_{01}^\obs,&         N_{11}-n_{01}^\obs \leq N-N_{10}-n_{10}^\obs-n_{01}^\obs,\\
N_{01}+N_{11}-n_{10}^\obs-n_{01}^\obs \leq N_{11},&         N_{01}+N_{11}-n_{10}^\obs-n_{01}^\obs \leq n_{11}^\obs, \\    
N_{01}+N_{11}-n_{10}^\obs-n_{01}^\obs \leq N_{01}+N_{11}-n_{01}^\obs,&         N_{01}+N_{11}-n_{10}^\obs-n_{01}^\obs\leq N-N_{10}-n_{10}^\obs-n_{01}^\obs .
\end{array}
$$
These inequalities can be simplied to be \eqref{eq::feasible}.
The inequality for $N_{01}$ is $N_{01}\leq n_{10}^\obs + n_{01}^\obs$, redundant over the sensitivity analysis region $N_{01}\leq N\widehat{p}_0(1-\widehat{p}_1)$, because $n_{10}^\obs + n_{01}^\obs\geq N\widehat{p}_0(1-\widehat{p}_1).$ 
\end{proof}

We next prove the propositions.  These proofs rely on the following lemma:

\begin{lemma}
\label{lemma::variance}
Assume $(c_1, \ldots, c_N)$ are constants with $\bar{c}=\sumN c_i/N$ and $S_c^2 = \sumN (c_i-\bar{c})^2/(N-1)$. Let $(W_1,\ldots, W_N)$ be the treatment indicators of a completely randomized experiment. We have that
$$
E\left(  \sumN W_ic_i  \right) = N_1 \bar{c},\quad
\var\left(  \sumN W_ic_i \right) = \frac{N_1N_0}{N}S_c^2. 
$$
\end{lemma}
See classical survey sampling textbooks \citep[e.g.,][]{cochran::1977} for the proof.

\begin{proof}
[Proof of Proposition \ref{thm::identi-mono}]
Verify that 
$$
E(n_{01}^\obs) = E\left\{  \sumN (1-W_i) Y_i(0)  \right\} = \frac{N_0}{N} N_{11},\quad 
E(n_{10}^\obs) = E\left[  \sumN W_i \{1-Y_i(1)\}  \right] = \frac{N_1}{N} N_{00}.
$$
The conclusion follows.
\end{proof}

\begin{proof}
[Proof of Proposition \ref{thm::moments-mono}]
Following \cite{neyman::1923} (presented using modern notation in \citet{imbens::2015book}), $\widehat{\tau}$ is unbiased for $\tau  $ with variance
\begin{eqnarray}
\label{eq::variance-neyman}
\var( \widehat{\tau}   ) = \frac{S_1^2}{N_1} + \frac{S_0^2}{N_0} - \frac{S_\tau^2}{N},
\end{eqnarray}
where 
\begin{eqnarray*}
S_1^2 &=& \frac{1}{N-1} \sumN \{ Y_i(1) - p_1 \}^2 = \frac{1}{N-1}  (Np_1^2-Np_1) = \frac{N}{N-1} p_1(1-p_1),\\
S_0^2 &=& \frac{1}{N-1} \sumN \{ Y_i(0) - p_0 \}^2 = \frac{1}{N-1}  (Np_0^2-Np_0) = \frac{N}{N-1} p_0(1-p_0),\\  
S_\tau^2&=&\frac{1}{N-1} \sumN (  \tau_i-\tau  )^2 = \frac{1}{N-1} \left(N_{10}  - \frac{N_{10}^2}{N}  \right)  = \frac{N}{N-1} \tau  (1-\tau  )
\end{eqnarray*}
are the finite population variance of $Y(1), Y(0)$, and $\tau.$ 
For estimating the variance, note that the variance term $S_\tau^2$ is identifiable because $N_{01} =0$ under monotonicity, and the conclusion follows.

The consistency and asymptotic normality of $\widehat{\tau}$ follows from the finite population central limit theorem \citep{hajek::1960}. And the variance estimator can be obtained by a simple plug-in. 
\end{proof}

\begin{proof}
[Proof of Proposition \ref{thm::identifiability-no-mono}]
From Lemma \ref{lemma::variance}, we have $E(\widehat{p}_1)=p_1$ and $E(\widehat{p}_0)=p_0$. Then  
\begin{eqnarray*}
E\left(  \frac{N}{N_0} n_{01}^\obs - N_{01} \right)&=& 
E\left(  N\widehat{p}_0 - N_{01} \right) = Np_0-N_{01}=(N_{01}+N_{11})-N_{01}=N_{11},\\
E\left(  \frac{N}{N_1} n_{10}^\obs - N_{01}   \right)&=&
E\{  N(1-\widehat{p}_1)  - N_{01} \}=N(1-p_1)-N_{01}=(N_{01}+N_{00})-N_{01}=N_{00},\\
E\left(  N+N_{01}-\frac{N}{N_0}n_{01}^\obs - \frac{N}{N_1}n_{10}^\obs    \right)&=&
N+N_{01}-Np_0-N(1-p_1)=N_{10}.
\end{eqnarray*}
\end{proof}

\begin{proof}
[Proof of Proposition \ref{thm::bounds-sensitivity}]
As a byproduct of the derivations in the proof of Proposition \ref{thm::identifiability-no-mono}, we have
\begin{eqnarray*}
Np_0-N_{01} \geq 0,\quad 
N(1-p_1)-N_{01} \geq 0,\quad 
N+N_{01}-Np_0-N(1-p_1)\geq 0,
\end{eqnarray*}
which further implies 
$
\max(0, - N\tau  ) \leq N_{01} \leq \min\{   Np_0, N(1-p_1)    \} .
$

Yule's measure of the correlation is $N_{11}N_{00} - N_{10}N_{11}$, which is the rescaled covariance of potential outcomes.  
If this is non-negative, the correlation of potential outcomes is non-negative.
We also have that $N_{00} = N(1-p_1) - N_{01}$, $N_{11} = Np_0 - N_{01}$, and $N_{10} = N+N_{01}-N(1-p_1) - Np_0,$ giving
\begin{eqnarray*}
0&\leq& N_{11}N_{00}-N_{10}N_{01} 
=  (Np_0 - N_{01}) \{ N(1-p_1) - N_{01} \} -  \{ N+N_{01}-N(1-p_1) - Np_0 \} N_{01} \\
&=&  N^2p_0(1-p_1) - NN_{01}, 
\end{eqnarray*}
or, equivalently, $N_{01} \leq Np_0(1-p_1).$
\end{proof}

\begin{proof}
[Proof of Proposition \ref{thm::variance-bounds-non-mono}]
According to the variance formula of $\widehat{\tau}  $ in (\ref{eq::variance-neyman}), we need to calculate $S_\tau^2/N$ with a known $N_{01}$. We have
\begin{eqnarray*}
\frac{S_\tau^2}{N} &=& \frac{1}{(N-1)N} \left(  \sumN \tau_i^2 - N\tau^2  \right) 
= \frac{1}{N-1} \left\{   \frac{N_{10}+N_{01}}{N}  -  \left(   \frac{ N_{10} - N_{01} }{N}  \right)^2   \right\}  
= \frac{1}{N-1} \left(\tau - \tau^2 + \frac{2N_{01} }{ N} \right ),
\end{eqnarray*}
and its bounds follows directly from $0\leq N_{01} \leq Np_0(1-p_1)$.
\end{proof}

\begin{proof}
[Proof of Proposition \ref{thm::A-neyman}]
We have $E(A) =  N_1\tau   = E(N_1\widehat{\tau}  )$, and 
\begin{eqnarray*}
\var(A - N_1\widehat{\tau}   )&=& \var\left[       \sumN W_i \{  Y_i(1)-Y_i(0)  \} - \sumN W_i Y_i(1) + \frac{N_1}{N_0}\sumN (1-W_i) Y_i(0)    \right]\\
&=&\var\left[      \sumN W_i  \left\{    Y_i(1)-Y_i(0) - Y_i(1) -\frac{N_1}{N_0} Y_i(0)     \right\}     \right]\\
&=& \var\left\{   \sumN  W_i Y_i(0) \cdot \frac{N}{N_0}    \right\}
= \frac{N^2}{N_0^2} \cdot \frac{N_1 N_0}{  N(N-1) } \cdot \sumN \{ Y_i(0) - \bar{Y}(0) \}^2 \\
&=&  \frac{N N_1 }{N_0}  S_0^2 =  \frac{N^2 N_1}{N_0(N-1)} p_0 (1-p_0),
\end{eqnarray*}
where the penultimate line of the proof is due to Lemma \ref{lemma::variance}.
\end{proof}

\section*{Acknowledgements}
The authors thank Dr. Avi Feller at Berkeley for helpful comments.

%

\bibliographystyle{biometrika}
\bibliography{ACE}

\begin{thebibliography}{35}
\expandafter\ifx\csname natexlab\endcsname\relax\def\natexlab#1{#1}\fi

\bibitem[{Agresti(2013)}]{agresti2013categorical}
\textsc{Agresti, A.} (2013).
\newblock \textit{Categorical Data Analysis}.
\newblock Hoboken, New Jersy: John Wiley \& Sons, Inc., 3rd ed.

\bibitem[{Aronow et~al.(2014)Aronow, Green \& Lee}]{aronow::2014}
\textsc{Aronow, P.~M.}, \textsc{Green, D.~P.} \& \textsc{Lee, D.~K.} (2014).
\newblock Sharp bounds on the variance in randomized experiments.
\newblock \textit{The Annals of Statistics} \textbf{42}, 850--871.

\bibitem[{Chiba(2015)}]{chiba2015exact}
\textsc{Chiba, Y.} (2015).
\newblock Exact tests for the weak causal null hypothesis on a binary out come
  in randomized trials.
\newblock \textit{Journal of Biometrics \& Biostatistics} \textbf{6}, 244.
  doi:10.4172/21556180.1000244.

\bibitem[{Cochran(1977)}]{cochran::1977}
\textsc{Cochran, W.~G.} (1977).
\newblock \textit{Sampling Techniques}.
\newblock John Wiley \& Sons: New York, 2nd ed.

\bibitem[{Copas(1973)}]{copas::1973}
\textsc{Copas, J.} (1973).
\newblock Randomization models for the matched and unmatched 2$\times$ 2
  tables.
\newblock \textit{Biometrika} \textbf{60}, 467--476.

\bibitem[{Ding(2014)}]{ding::2014}
\textsc{Ding, P.} (2014).
\newblock A paradox from randomization-based causal inference.
\newblock \textit{arXiv:1402.0142} .

\bibitem[{Ding \& Dasgupta(2016)}]{ding::2015}
\textsc{Ding, P.} \& \textsc{Dasgupta, T.} (2016).
\newblock A potential tale of two by two tables from completely randomized
  experiments.
\newblock \textit{Journal of the American Statistical Association} , DOI:
  10.1080/01621459.2014.995796.

\bibitem[{Fisher(1935)}]{fisher::1935}
\textsc{Fisher, R.~A.} (1935).
\newblock \textit{The Design of Experiments.}
\newblock Edinburgh: Oliver \& Boyd, 1st ed.

\bibitem[{Fogarty et~al.(2016)Fogarty, Mikkelsen, Gaieski \&
  Small}]{fogarty2016discrete}
\textsc{Fogarty, C.~B.}, \textsc{Mikkelsen, M.~E.}, \textsc{Gaieski, D.~F.} \&
  \textsc{Small, D.~S.} (2016).
\newblock Discrete optimization for interpretable study populations and
  randomization inference in an observational study of severe sepsis mortality.
\newblock \textit{Journal of the American Statistical Association} ,
  DOI:10.1080/01621459.2015.1112802.

\bibitem[{Greenland(1991)}]{greenland1991logical}
\textsc{Greenland, S.} (1991).
\newblock On the logical justification of conditional tests for two-by-two
  contingency tables.
\newblock \textit{The American Statistician} \textbf{45}, 248--251.

\bibitem[{H{\'a}jek(1960)}]{hajek::1960}
\textsc{H{\'a}jek, J.} (1960).
\newblock Limiting distributions in simple random sampling from a finite
  population.
\newblock \textit{Publications of the Mathematics Institute of the Hungarian
  Academy of Science} \textbf{5}, 361--74.

\bibitem[{Hansen \& Bowers(2009)}]{hansen2009attributing}
\textsc{Hansen, B.~B.} \& \textsc{Bowers, J.} (2009).
\newblock Attributing effects to a cluster-randomized get-out-the-vote
  campaign.
\newblock \textit{Journal of the American Statistical Association}
  \textbf{104}, 873--885.

\bibitem[{Hodges \& Lehmann(1963)}]{hodges1963estimates}
\textsc{Hodges, J. J.~L.} \& \textsc{Lehmann, E.~L.} (1963).
\newblock Estimates of location based on rank tests.
\newblock \textit{The Annals of Mathematical Statistics} \textbf{34}, 598--611.

\bibitem[{Imbens \& Rosenbaum(2005)}]{imbens2005robust}
\textsc{Imbens, G.~W.} \& \textsc{Rosenbaum, P.~R.} (2005).
\newblock Robust, accurate confidence intervals with a weak instrument: quarter
  of birth and education.
\newblock \textit{Journal of the Royal Statistical Society: Series A
  (Statistics in Society)} \textbf{168}, 109--126.

\bibitem[{Imbens \& Rubin(2015)}]{imbens::2015book}
\textsc{Imbens, G.~W.} \& \textsc{Rubin, D.~B.} (2015).
\newblock \textit{Causal Inference in Statistics, Social, and Biomedical
  Sciences: An Introduction}.
\newblock New York: Cambridge University Press.

\bibitem[{Keele et~al.(2015)Keele, Small \& Grieve}]{keele2015randomization}
\textsc{Keele, L.}, \textsc{Small, D.} \& \textsc{Grieve, R.} (2015).
\newblock Randomization based instrumental variables methods for binary
  outcomes with an application to the improve trial.
\newblock Tech. rep., Penn State University.

\bibitem[{Li \& Ding(2016)}]{li2015exact}
\textsc{Li, X.} \& \textsc{Ding, P.} (2016).
\newblock Exact confidence intervals for the average causal effect on a binary
  outcome.
\newblock \textit{Statistics in Medicine} \textbf{35}, 957--960.

\bibitem[{Nelsen(2007)}]{nelsen2007introduction}
\textsc{Nelsen, R.~B.} (2007).
\newblock \textit{An Introduction to Copulas}.
\newblock New York: Springer, 2nd ed.

\bibitem[{Neyman(1923)}]{neyman::1923}
\textsc{Neyman, J.} (1923).
\newblock On the application of probability theory to agricultural experiments.
  {E}ssay on principles. {S}ection 9.
\newblock \textit{Statistical Science} \textbf{5}, 465--472.

\bibitem[{Nolen \& Hudgens(2011)}]{nolen2011randomization}
\textsc{Nolen, T.~L.} \& \textsc{Hudgens, M.~G.} (2011).
\newblock Randomization-based inference within principal strata.
\newblock \textit{Journal of the American Statistical Association}
  \textbf{106}, 581--593.

\bibitem[{Rigdon \& Hudgens(2015{\natexlab{a}})}]{rigdon2015exact}
\textsc{Rigdon, J.} \& \textsc{Hudgens, M.~G.} (2015{\natexlab{a}}).
\newblock Exact confidence intervals in the presence of interference.
\newblock \textit{Statistics and Probability Letters} \textbf{105}, 130--135.

\bibitem[{Rigdon \& Hudgens(2015{\natexlab{b}})}]{rigdon2015randomization}
\textsc{Rigdon, J.} \& \textsc{Hudgens, M.~G.} (2015{\natexlab{b}}).
\newblock Randomization inference for treatment effects on a binary outcome.
\newblock \textit{Statistics in Medicine} \textbf{34}, 924--935.

\bibitem[{Robins(1988)}]{robins::1988}
\textsc{Robins, J.~M.} (1988).
\newblock Confidence intervals for causal parameters.
\newblock \textit{Statistics in Medicine} \textbf{7}, 773--785.

\bibitem[{Rosenbaum(1999)}]{rosenbaum1999reduced}
\textsc{Rosenbaum, P.~R.} (1999).
\newblock Reduced sensitivity to hidden bias at upper quantiles in
  observational studies with dilated treatment effects.
\newblock \textit{Biometrics} \textbf{55}, 560--564.

\bibitem[{Rosenbaum(2001)}]{rosenbaum::2001}
\textsc{Rosenbaum, P.~R.} (2001).
\newblock Effects attributable to treatment: Inference in experiments and
  observational studies with a discrete pivot.
\newblock \textit{Biometrika} \textbf{88}, 219--231.

\bibitem[{Rosenbaum(2002)}]{rosenbaum::2002}
\textsc{Rosenbaum, P.~R.} (2002).
\newblock \textit{Observational Studies}.
\newblock Springer: New York, 2nd ed.

\bibitem[{Rosenbaum(2012)}]{rosenbaum2012interference}
\textsc{Rosenbaum, P.~R.} (2012).
\newblock Interference between units in randomized experiments.
\newblock \textit{Journal of the American Statistical Association}
  \textbf{102}, 191--200.

\bibitem[{Rubin(1974)}]{rubin::1974}
\textsc{Rubin, D.~B.} (1974).
\newblock Estimating causal effects of treatments in randomized and
  nonrandomized studies.
\newblock \textit{Journal of Educational Psychology} \textbf{66}, 688--701.

\bibitem[{Rubin(1980)}]{rubin::1980}
\textsc{Rubin, D.~B.} (1980).
\newblock Comment on ``{R}andomization analysis of experimental data: {T}he
  {F}isher randomization test''.
\newblock \textit{Journal of the American Statistical Association} \textbf{75},
  591--593.

\bibitem[{Rubin(1984)}]{rubin::1984}
\textsc{Rubin, D.~B.} (1984).
\newblock Bayesianly justifiable and relevant frequency calculations for the
  applies statistician.
\newblock \textit{The Annals of Statistics} \textbf{12}, 1151--1172.

\bibitem[{Rubin(1998)}]{rubin::1998}
\textsc{Rubin, D.~B.} (1998).
\newblock More powerful randomization-based $p$-values in double-blind trials
  with non-compliance.
\newblock \textit{Statistics in Medicine} \textbf{17}, 371--385.

\bibitem[{Rubin(2005)}]{rubin::2005}
\textsc{Rubin, D.~B.} (2005).
\newblock Causal inference using potential outcomes: {D}esign, modeling, and
  decisions.
\newblock \textit{Journal of the American Statistical Association}
  \textbf{100}, 322--331.

\bibitem[{Wang(2015)}]{wang2015exact}
\textsc{Wang, W.} (2015).
\newblock Exact optimal confidence intervals for hypergeometric parameters.
\newblock \textit{Journal of the American Statistical Association}
  \textbf{110}, 1491--1499.

\bibitem[{Watson(2014)}]{watson2014complications}
\textsc{Watson, D.~A.} (2014).
\newblock \textit{Complications in Causal Inference: Incorporating Information
  Observed After Treatment is Assigned}.
\newblock Ph.D. thesis, Department of Statistics, Harvard University.

\bibitem[{Yates(1984)}]{yates1984tests}
\textsc{Yates, F.} (1984).
\newblock Tests of significance for $2 \times 2$ contingency tables (with
  discussion).
\newblock \textit{Journal of the Royal Statistical Society. Series A (General)}
  \textbf{147}, 426--463.

\end{thebibliography}

%

\end{document}